\documentclass[journal,twocolumn]{IEEEtran}
\usepackage{xcolor}
\usepackage{amsmath,amsthm,verbatim,amssymb,amsfonts,amscd, graphicx,algorithm,hyperref,tabularx,microtype,bm}
\usepackage{mathtools}
\usepackage[noend]{algpseudocode}
\usepackage[inline]{enumitem}
\theoremstyle{plain}
\newtheorem{theorem}{Theorem}

\newtheorem{lemma}{Lemma}

\theoremstyle{definition}
\newtheorem{definition}{Definition}
\newtheorem{construction}{Construction}

\newtheorem{example}{Example}

\newcolumntype{C}{>{\centering\arraybackslash}p{6em}}

\def\C{\mathcal{C}}
\def\rk{\text{rank}}
\def\F{\mathbb{F}}
\def\te{\tilde{e}}

\def\a{\alpha}
\def\g{\gamma}

\usepackage{cite}
\begin{document}
\title{Irregular Recovery and Unequal Locality for Locally Recoverable Codes with Availability}
\author{Sourbh~Bhadane~and~Andrew~Thangaraj,~\IEEEmembership{Senior~Member,~IEEE,}\\
\thanks{S. Bhadane and A. Thangaraj are with the Department of Electrical Engineering,
Indian Institute of Technology Madras, Chennai 60036, India. Email: sourbh.bhadane,andrew@ee.iitm.ac.in.}}

\maketitle

\begin{abstract}
A code is said to be a Locally Recoverable Code (LRC) with availability if every coordinate can be recovered from multiple disjoint sets of other coordinates called recovering sets. The vector of sizes of recovering sets of a coordinate is called its recovery profile. In this work, we consider LRCs with availability under two different settings: (1) \emph{irregular recovery}: non-constant recovery profile that remains fixed for all coordinates, (2) \emph{unequal locality}: regular recovery profile that can vary with coordinates. For each setting, we derive bounds for the minimum distance that generalize previously known bounds to the cases of irregular or varying recovery profiles. For the case of regular and fixed recovery profile, we show that a specific Tamo-Barg polynomial-evaluation construction is optimal for all-symbol locality, and we provide parity-check matrix constructions for information locality with availability. 
\end{abstract}

\section{Introduction}
Modern distributed storage systems that store a large amount of data are prone to node failures. Replication has been a popular and traditional method for protecting against failures and providing reliability. Recently, instead of replication of data, erasure codes have been employed to reduce storage overhead, while maintaining the same level of reliability. For example, Facebook uses a $\left(14,10 \right)$ Reed-Solomon (RS) code instead of replication that can recover from as many as 4 node failures. However, an RS decoder needs to read from 10 other nodes for the recovery of even a single node failure. Since a single node failure is the most frequent and since reading from fewer nodes for recovery is desirable, researchers have proposed Locally Recoverable Codes (LRCs) \cite{gopalan2012locality}. LRCs were originally intended to minimize the number of nodes accessed to recover from a single node failure. Although single node failures are most frequent, LRCs with  multiple disjoint \textit{recovering sets} are useful for recovering from multiple concurrent node failures. Moreover, this property could also be exploited for the storage of ``hot'' data, which may be served to several users simultaneously using the recovering sets in parallel. As a result, LRCs with multiple disjoint recovering sets are also referred to as LRCs with \textit{availability}.  \par
A coordinate is said to have locality $r$ if it can be recovered by accessing at most $r$ other coordinates. An LRC is said to have \textit{information locality} $r$ if all information coordinates have locality $r$. If all coordinates have locality $r$, an LRC is said to have \textit{all-symbol locality}. LRCs were first introduced in the seminal paper \cite{gopalan2012locality} and a Singleton-like upper bound on the minimum distance was derived. Constructions meeting this bound with exponential field size were proposed in \cite{6620541}, \cite{6620540}, \cite{DBLP:journals/corr/HaoX16}. A parity-check matrix approach was used in \cite{DBLP:journals/corr/HaoX16} to construct optimal LRC codes. An elegant algebraic optimal construction with field size linear in blocklength was presented in \cite{tamo2014family}. LRCs with availability have been studied in \cite{tamo2014family,6620355,rawat2014locality,wang2014repair,WangZL15,anyu2014combi,tamo2016bounds,huang2015linear}. In \cite{6620355}, LRCs with availability were first introduced and explicit constructions using partial geometries were presented. The algebraic construction in \cite{tamo2014family} includes extensions to LRCs with availability. In \cite{anyu2014combi}, repair groups were related to a combinatorial concept of regenerating sets and minimum distance bounds for square codes were developed. In \cite{rawat2014locality}, minimum distance upper bounds were derived and optimal constructions using Gabidulin codes were presented for a weaker notion of LRCs with availability. High-rate constructions of binary LRCs with availability using block designs were presented in \cite{WangZL15}. A field-size dependent distance upper bound for linear LRCs with information locality and availability was derived in \cite{huang2015linear} and a tensor product based code was constructed to achieve optimality in some cases. The following upper bound on the minimum distance for an $\left[ n,k,d \right] $ LRC with information locality $r$ and availability $t$ is given in \cite{wang2014repair}:
\begin{equation} \label{infobd}
d \leq n - k - \left \lceil \frac{t \left(k-1 \right) +1}{t \left( r -1 \right) +1} \right \rceil + 2.
\end{equation}
For $n \geq k \left(tr+1 \right)$, \cite{wang2014repair} proves existence of codes that meet the above bound. However, to the best of our knowledge, no explicit constructions meeting this bound are known. In \cite{tamo2016bounds}, the following upper bound on minimum distance was derived for the all-symbol locality and availability case:
\begin{equation} \label{tbbound}
d \leq n- \sum\limits_{i=0}^{t} \left \lfloor \frac{ k-1}{r^i} \right \rfloor.
\end{equation}
To the best of our knowledge, no general constructions that attain the above bound are known. Recently, \cite{7541336}, \cite{7541377} studied codes with unequal locality. Upper bounds on minimum distance for codes with unequal information locality and unequal all-symbol locality were obtained \cite{7541336}. Constructions based on an adaptation of Pyramid codes and rank-metric codes were proposed to attain these bounds, respectively. However, \cite{7541336}, \cite{7541377} did not consider LRCs with availability.\par
In an LRC with availability, the sizes of recovering sets of a particular coordinate is called its recovery profile, which is said to be regular if all sizes are equal, and irregular otherwise. In this work, we extend codes with availability to include irregular and varying recovery profiles. Specifically, we study the following two settings, which do not appear to have been studied in any of the mentioned prior work:
\begin {enumerate} [label=\itshape\alph*\upshape)]
\item \emph{Irregular recovery}: recovery profile can have varying recovering set sizes but remains fixed for all coordinates, i.e. $t$ disjoint recovering sets with sizes $r_1$, $r_2$, $\ldots$, $r_t$ for all coordinates, 
\item \emph{Unequal locality}: regular recovery profile that may vary over coordinates, i.e. $t$ disjoint recovering sets each of size $r_i$ for coordinate $i$, $1\le i\le n$.
\end {enumerate} 
Upper bounds on minimum distance are obtained for both settings under information and/or all-symbol locality. We also present a generalization of an existing construction from \cite{tamo2014family} and prove that it meets \eqref{tbbound} for arbitrary $t$ and $r=k-1$. For information locality and availability, we extend the construction in \cite{{DBLP:journals/corr/HaoX16}} to include availability and provide an explicit parity-check matrix construction that meets \eqref{infobd} for $n \geq k \left(tr+1 \right)$. \thanks{An earlier version of this work was presented partly in the National Conference on Communications, IIT Madras, Mar 2017.}

\section{Preliminaries}
Consider an $\left[ n,k,d \right]$ linear code $\mathcal{C} \subseteq \mathbb{F}_{q}^{n}$, where $q$ is a prime power and $\F_q$ is the finite field with $q$ elements. Suppose a subset $D\subseteq \left[ n \right] \triangleq \{1,2,\ldots,n\}$ is the support of a dual codeword. Then, for every $i\in D$, the $i$-th coordinate of a codeword of $\C$ is a linear combination of the coordinates in $D\setminus\{i\}$. The code $\C$ is said to have locality $r$ and availability $t$ if, for $i \in [n]$, there exist $t$ dual-codeword support sets $D_j^{\left( i \right)}$, $j \in \left[ t \right]$ such that (1) $i\in\left|D_j^{(i)}\right|$, (2) $R_j^{(i)}=D_j^{(i)}\setminus\{i\}$ are disjoint and (3) $\left|R_j^{(i)}\right|\le r$. The sets $R_j^{(i)}$ are called the \emph{recovering sets} for $i$ because the coordinate $i$ can be recovered from the coordinates in any of its recovering sets. For $i\in[n]$, we denote $\Gamma_{a} \left( i \right) = \left\lbrace i \right\rbrace \cup R_{1}^{\left( i \right)} \cup R_{2}^{\left( i \right)} \cdots \cup R_{a}^{\left( i \right) }$ for $1 \leq a \leq t $. A code with locality is referred to as a Locally Recoverable Code (LRC). 

\subsection{Minimum distance bound and algorithm}
For $S\subseteq [n]$, let $\C_S$ denote the code $\C$ restricted to the positions in $S$, and let $\rk(S)$ denote the dimension of $\C_S$. A useful bound on minimum distance of $\C$ is the following: if $\rk(S)<k$, then $d\le n-|S|$. For LRCs, Algorithm \ref{alg:1} is typically used in proofs of minimum distance bounds to find a set $S$ for which $\rk(S)<k$ \cite{gopalan2012locality}\cite{wang2014repair}. 
 
\begin{algorithm}[htb]
\caption{\text{Construct $S$ such that $\text{rank} \left( S \right) = k-1$ }}
\label{alg:1}
\begin{algorithmic}[1]
\State Set  $S_{0}=\phi$, $i = 0$
\While {$\text{rank} \left( S_i \right) \leq k-2}$
\State Set $i=i+1$, Choose $j \in [n] \setminus S_{i-1}$
\If {$\text{rank} \left( S_{i-1} \cup \Gamma_{t} \left( j \right) \right) < k $} 
\State Set $S_{i} = S_{i-1} \cup \Gamma_{t}  \left( j \right) $ 
\Else 
\State Choose $a$ s. t. $\text{rank} \left( S_{i-1} \cup \Gamma_{a+1} \left( j \right) \right) = k$ and
\State $R \subseteq R^{\left(j \right)}_{a+1}$ s. t. $\text{rank} \left( S_{i-1} \cup \Gamma_{a} \left( j \right) \cup R \right) = k-1$
\State Set $S_i = S_{i-1} \cup \Gamma_{a} \left( j \right) \cup R $
\EndIf	
\EndWhile
\State Return $S=S_i$
\end{algorithmic}
\end{algorithm}

\subsection{Tamo-Barg polynomial-evaluation construction}
\label{sec:polyn-eval-constr}
We describe the polynomial-evaluation construction of LRCs with availability from \cite{tamo2014family}.

Let $A\subseteq \mathbb{F}$ ($\F$ is a finite field), $|A|=n$. Let $\mathcal{A}_1$ and $\mathcal{A}_2$ be two partitions of $A$ such that for any two sets $A_1 \in \mathcal{A}_1$, $A_2 \in \mathcal{A}_2$, we have $|A_1|=r_1$, $|A_2|=r_2$, and the size of their intersection $|A_1 \cap A_2| \leq 1$. Such partitions are called orthogonal partitions. 
Define for $\mathcal{A}_i$, $i=1,2$,
\begin{align*}
\mathbb{F}_{\mathcal{A}_i}[x] &= \left \{ f \in \mathbb{F}[x] : \right. \\
 f & \left. \text{ is constant on } A_i \in \mathcal{A}_i, \text{deg }f \leq |A| \right \}.
\end{align*}
Further, define two families of polynomials
\[ \mathcal{F}_{\mathcal{A}_1}^{r_1} = \oplus_{i=0}^{r_1-1} \mathbb{F}_{\mathcal{A}_1} [x] x^i, \hspace{5mm}
\mathcal{F}_{\mathcal{A}_2}^{r_2} = \oplus_{i=0}^{r_2-1} \mathbb{F}_{\mathcal{A}_2} [x] x^i. \]
Consider a polynomial $f$ that belongs to the intersection of $\mathcal{F}_{\mathcal{A}_1}^{r_1}$ and $\mathcal{F}_{\mathcal{A}_2}^{r_2}$. A codeword of a length-$n$ LRC with availability $t=2$ is obtained by evaluating $f$ on all $n$ points of $A$. If the number of such polynomials of degree at most $m$ is $|\mathbb{F}|^k$, we obtain an $(n,k,d)$ availability-2 LRC with minimum distance $d\ge n-m$. A set $A_i\in\mathcal{A}_i$ is a dual codeword support set because the $r_i$ points of $A_i$ pass through a polynomial of degree at most $r_i-1$.

Following \cite{tamo2014family}, a partition is naturally formed by a subgroup $H$ of the multiplicative or additive group of $\mathbb{F}$ and cosets of $H$. A degree-$|H|$ polynomial constant on such partitions is 
\[ g \left( x \right) = \prod_{h \in H} \left( x-h \right). \] 
Such a polynomial is called the annihilator polynomial of $H$.
If $H$ is a multiplicative subgroup of $\mathbb{F}_q^{*}$, then $g \left( x \right) = x^{|H|}$ is constant on each coset of $H$.
\section{Irregular Recovery with Availability}
\label{section:3}
In this section, we consider locally recoverable codes (LRCs) whose coordinates have an irregular recovery profile. This extends the notion of $\left( r,t \right) $ locality in %\cite{DBLP:journals/tit/WangZ14}, 
 \cite{wang2014repair} to the case where sizes of recovering sets of each coordinate are not equal. A precise definition is as follows.

\begin{definition}
Let $\mathcal{C} \subseteq \mathbb{F}_{q}^{n}$ be an $\left[ n,k,d \right] $ code. The $i$-th coordinate has $\left( \bm{r},t \right) $ locality, where $\bm{r} = \left( r_{1}, r_{2} \dots r_{t} \right)$, if there are $t$ disjoint recovering sets $R_{1}^{\left( i\right) }, R_{2}^{\left( i \right) } \dots R_{t}^{\left( i \right)}$ for $i$ such that 
\[ \left|R_{j}^{\left( i \right)}\right| \le r_{j} \hspace{5mm} \forall j \in \left[ t \right]. \]
The code $\C$ has $(\bm{r},t)$ information locality if all information coordinates have $(\bm{r},t)$ locality. The code $\C$ has $(\bm{r},t)$ all-symbol locality if all coordinates have $(\bm{r},t)$ locality. 
\end{definition}

\subsection{Information locality}
First we consider bounds on minimum distance of codes with $(\bm{r},t)$ information locality. We follow a similar proof technique as \cite{gopalan2012locality} \cite{wang2014repair} but with adaptations for unequal recovery.
\begin{theorem} \label{th1}
If $\mathcal{C}$ has $\left( \bm{r}, t \right)$ information locality, then
$$ d \leq n - k - \left \lceil{\frac{t(k-1)+1}{ \sum_{j=1}^{t} \left( r_{j} - 1\right) +1 }} \right \rceil +2 $$
\end{theorem} 
\begin{IEEEproof}
In  the proof, we will assume that $r_1\le r_2\le\cdots\le r_t$. Let $l$ be the number of iterations in Algorithm 1 when run on the code $\C$ resulting in a subset $S$. 
Denote the rank increment and size increment in the $i$-th iteration of Algorithm 1 by $m_{i}=\rk(S_i)-\rk(S_{i-1})$ and $s_{i}=|S_i|-|S_{i-1}|$, respectively. We find a lower bound for $|S|$, which in turn leads to an upper bound on the distance $d\le n-|S|$.   
Consider two cases depending on how Algorithm 1 terminates.

\noindent \emph{Case 1}: $S_l$ is a union of $\Gamma_{t} \left( j \right)$'s i.e, $S_{l}$ is formed in line 5. 

Since each of the recovering sets contribute at least one linear dependency to $S_{i}$,   
 we have $ m_{i} \leq s_{i} -t$ for $i \in \left[ l \right]$. Further, we have  
 \begin{align}\label{expS}
 |S| = \sum\limits_{i=1}^{l} s_{i} \geq \sum\limits_{i=1}^{l} \left( m_{i} + t \right)= k-1 + tl.
\end{align}
To find a lower bound on $|S|$, we find a lower bound on $l$, the number of iterations. Since every recovering set adds at least 1 linear equation, we have 
$$ \text{rank} \left( \Gamma _{a} \left( j \right) \right) \leq 1 + \sum\limits_{j'=1}^{a} \left( r_{j'} - 1 \right). $$
Since $S$ is the union of $l$ sets $\Gamma_t(j)$, we have
\begin{align}\label{eq:llb} 
\text{rank} \left( S \right) = k-1 &\leq l \big( 1 + \sum\limits_{j=1}^{t} \left( r_{j} - 1 \right) \big). 
\end{align}  
Using \eqref{eq:llb} in \eqref{expS} and $d\le n-|S|$, we get
\begin{align}
d &\leq n - \left( k-1 + t \left \lceil {\frac{k-1}{1+\sum_{j=1}^{t} \left( r_j -1 \right)}} \right \rceil \right) \nonumber \\
&\leq n - k - t\left \lceil {\frac{ k-1  }{1+\sum_{j=1}^{t} \left( r_j -1 \right)}} \right \rceil +1\label{eq:6}\\
&\leq n - k - \left \lceil {\frac{t \left( k-1 \right) +1 }{1+\sum_{j=1}^{t} \left( r_j -1 \right)}} \right \rceil +2,\label{eq:7}
\end{align}
where, to get from \eqref{eq:6} to \eqref{eq:7}, we use the facts $\lfloor tx\rfloor\le tx\le t\lceil x\rceil$ for a real number $x$ and $\lfloor\frac{a}{b}\rfloor=\lceil\frac{a+1}{b}\rceil-1$ for positive integers $a$, $b$.

\noindent \emph{Case 2}: $S_{l}$ is formed in line 9. 

Since $m_i\le s_i-t$, $1\le i\le l-1$, and $m_l\le s_l-a$, we have 
\begin{align}\label{expS2}
 |S| = \sum\limits_{i=1}^{l} s_{i} &\geq \sum\limits_{i=1}^{l-1} \left( m_{i} + t \right) + m_l + a \nonumber \\
 &= k-1 + t \left( l-1 \right) + a.
\end{align} 
Since $\text{rank} \left( S_{l-1} \cup \Gamma_{a+1} \left( l \right) \right) = k$, and $S_{l-1}$ is the union of $l-1$ sets $\Gamma_t(j)$, we have
\begin{align}
k \leq (l-1) \left( 1 + \sum\limits_{j=1}^{t} \left( r_{j} - 1 \right) \right)+\left(1+\sum_{j=1}^{a+1}(r_j-1)\right). \label{eq:1}
\end{align}
Using the lower bound for $l-1$ from \eqref{eq:1} in \eqref{expS2}, we get
\begin{align}
|S| &\geq k-1 + t \frac{k-1-\sum_{j=1}^{a+1}(r_j-1)}{1+ \sum_{j=1}^{t} \left( r_j - 1 \right)} + a.\label{eq:4}
\end{align} 
Let $\Omega=1+ \sum_{j=1}^{t} \left( r_j - 1 \right)$. Since the $r_j$ are in increasing order, the average of the first $a+1$ of the $(r_j-1)$ is smaller than the average of all $t$ resulting in the inequality
\begin{equation}
  \label{eq:3}
  \frac{\sum_{j=1}^{a+1}(r_j-1)}{a+1}\le \frac{\Omega-1}{t}.
\end{equation}
Using \eqref{eq:3} in \eqref{eq:4} and simplifying, we get
\begin{align}
  \label{eq:5}
  |S|&\ge k+\frac{t(k-1)+a+1}{\Omega}-2\\
&\ge k+\left\lceil\frac{t(k-1)+1}{\Omega}\right\rceil-2.
\end{align}
Using the above in $d\le n-|S|$, we get the statement of the theorem.
\end{IEEEproof} 
For the case of equal recovery with availability, $\bm{r}$ is a constant vector with $r_i=r$, and the bound of Theorem 1 reduces to \eqref{infobd}.

\subsection{All-symbol locality}
We now derive a minimum distance upper bound for codes with $\left( \bm{r},t \right)$ all-symbol locality. Our proof is similar in outline to\cite{tamo2016bounds}, and for equal recovery the bound reduces to \eqref{tbbound}. We present two lemmas required for the proof of the upper bound.  
For the rest of this section, we assume that $r_1 \leq r_2 \leq \cdots \le r_t$.  We use the notions of recovering graph and expansion ratio from \cite{tamo2016bounds}.
\subsubsection*{Recovering Graph and Expansion Ratio}
The recovering graph of a length-$n$  LRC code with $(\bm{r},t)$ locality has vertex set $\left[ n \right]$ and edges of color $j$ from vertex $i$ to $i'$ if $i' \in R_j^{(i)}$ for $j \in \left[ t \right]$. More generally, a $t$-edge-colored directed graph is said to be an $(\bm{r},t)$ recovering graph if, for $j\in[t]$, every vertex has at least one and at most $r_j$ outgoing color-$j$ edges. The set of vertices with incoming edges of the same color from a given vertex $i$ is said to be a recovery set for $i$. Note that the recovering graph of an LRC code with $(\bm{r},t)$ locality is indeed an $(\bm{r},t)$ recovering graph with $R^{(i)}_j$ being the recovery set of vertex $i$.

Consider an arbitrary subset of vertices $S$ in a recovering graph. Color all the vertices in $S$ in some fixed color, say red. Color every vertex with at least one fully colored recovery set. Continue this procedure until no more vertices can be colored. The final set of colored vertices thus obtained is defined to be the \textit{closure} of $S$, denoted by $\text{Cl} \left( S \right)$. The \textit{expansion ratio} with respect to $S$, denoted $e(S)$, is defined as the ratio $e(S)=|\text{Cl} \left( S \right)| / |S|$. Observe that all vertices in $\text{Cl} \left( S \right)$ can be recovered from vertices in $S$.

\begin{lemma}
 Let $G$ be an $(\bm{r},t)$ recovering graph. For a vertex $v \in G$, there exists a subset of the vertices $S$ such that $v \in \text{Cl} \left( S \right)$ and 
   \begin{equation}
     \label{eq:2}
|S|\leq \prod\limits_{i=1}^{t} r_i\text{ and }e(S)\ge e_{t}\triangleq 1 + \sum\limits_{j=1}^{t} \frac{1}{\prod\limits_{i=1}^{j} r_i}.     
   \end{equation}
\end{lemma}
\begin{IEEEproof}
Our proof follows the proof of Lemma 3 in \cite{tamo2016bounds} closely, except we use the following key insight to construct $S$: smaller recovering sets result in larger expansion ratios. Therefore, while constructing $S$ we give a higher preference to smaller recovering sets as compared to larger recovering sets. 

We proceed by induction on $t$. For $t=0$, we take $S$ as the single vertex $v$, and get $e(S)\ge1$. Making the induction hypothesis that the lemma is true for $t=k$, we consider the graph $G$ to be a $(\bm{r},k+1)$-recovering graph. From $G$, an $(\bm{r},k)$-recovering graph $G_1$ is first constructed by removing certain vertices and edges as follows: Remove vertex $v$ from $G$. For every other vertex, $u \ne v$, if there is a color-$j$ edge from $u$ to $v$, remove the edges corresponding to the $j$-th recovering set of $u$. If there is no edge from $u$ to $v$, remove the edges that correspond to the recovering set of $u$ with size $r_{1}$. We remove the $r_1$-sized recovering sets because of the aforementioned principle of giving higher preference to smaller sized recovering sets. It is easy to see that a vertex $u \in G_1$ has recovery profile $\widetilde{\bm{r}}=[\widetilde{r_1}, \widetilde{r_2},\ldots, \widetilde{r}_{k}]$, where $\widetilde{r_i} \leq r_{i+1}$. 

We now briefly outline the construction of the set $S$ as given in \cite{tamo2016bounds}. Let $v_1$, $v_2$,$\ldots$, $v_l$ be the vertices in the $r_{1}$-sized recovering set of $v$ in $G$, where $l \leq r_{1}$. By the induction hypothesis, there exists a subset $S_1\subseteq V(G_1)$ such that $v_1 \in \text{Cl}_1 \left( S_1 \right)$ ($\text{Cl}_i$ denotes closure in $G_i$), and $|S_1|$ and $e \left( S_1 \right)$ satisfy \eqref{eq:2} with recovery profile $\widetilde{\bm{r}}$. For $i=2$ to $l$, let $G_i$ be the induced subgraph defined on $V(G_i) \backslash \text{Cl}_i \left(\cup_{j=1}^{i-1} S_j \right)$. Every vertex in $G_i$ has at least one edge in each of its recovering set, since otherwise it would have been a part of $\text{Cl}_i \left(\cup_{j=1}^{i-1} S_j \right)$. Therefore, $G_i$ is a $(\widetilde{\bm{r}},k) $ recovering graph. If vertex $v_i$ is not in $G_i$, set $S_i = \phi$. If $v_i$ is in $G_i$, by the induction hypothesis, there exists a set $S_i$ in $G_i$ such that $v_i \in \text{Cl}_i(S_i)$ and $|S_i|, e \left(S_i \right)$ satisfy \eqref{eq:2} with recovery profile $\widetilde{\bm{r}}$. Let $S$ be the union of the sets $S_i$ for $1\le i\le l$. 

The upper bound on $|S|$ is immediate. Since $S=\cup_{i=1}^lS_i$ with $|S_i|\le \prod_{j=1}^k\widetilde{r_i}$ and $\widetilde{r_i} \leq r_{i+1}$, we have
\[ |S| \leq \sum\limits_{i=1}^{l} |S_i| \leq r_{1}  \prod\limits_{i=1}^{k} \widetilde{r_i} \leq \prod\limits_{i=1}^{k+1} r_i. \]  
For the lower bound on $e(S)$, note that $\text{Cl}_i(S_i)$, $1\le i\le l$, are disjoint, and that $e(S_i)\ge e_k$. So, we have
\begin{align*}
 e\left(S \right) &= \dfrac{|\text{Cl} \left(S \right)|}{|S|} = \dfrac{1+\sum\limits_{i=1}^{l} |\text{Cl}_i \left(S_i \right)|}{|S|} \\
 &\geq \frac{1}{\prod\limits_{i=1}^{k+1} r_i} + \sum\limits_{i=1}^{l-1} e(S_i) \dfrac{|S_i|}{|S|} \\
 &\geq \frac{1}{\prod\limits_{i=1}^{k+1} r_i} + e_k=e_{k+1}.
\end{align*}
\end{IEEEproof}
The increasing order $r_1\le\cdots\le r_t$ ensures that the largest possible expansion ratio can be obtained using $ \left \lbrace r_1, r_2, \ldots, r_t \right \rbrace$. Note that Lemma 1 reduces to Lemma 3 in \cite{tamo2016bounds} for the case of equal recovery when $r_i= r$.

The radix-$r$ representation of an integer plays a role in \cite{tamo2016bounds}. 
The analog for unequal recovery is the representation of an integer in the unequal radix $\{1,r_1,r_1r_2,\ldots,r_1r_2\cdots r_t\}$. The next lemma concerns such representations.
\begin{lemma}
Let $m$ be an integer with the following representation 
\[ m = \beta r_t \prod\limits_{i=1}^{t} r_i + \sum\limits_{i=1}^{t} \alpha_i \prod\limits_{j=1}^{i} r_j + \alpha_0, \] 
where $0\le\alpha_i< r_{i+1}$, $0\le i\le t-1$, $0\le\alpha_t<r_t$, and $\beta$ is an integer. Let $\te_i=1 + \sum\limits_{j=1}^{i} \frac{1}{\prod\limits_{l=1}^{j}r_l}$. 
Then,
\[ \Biggl \lfloor \frac{m}{\prod\limits_{i=1}^{t} r_i } \Biggr \rfloor e_t \prod_{i=1}^{t} r_i  + \sum\limits_{i=0}^{t-1} \alpha_i \te_i \prod_{j=1}^{i} r_j = \sum\limits_{i=0}^{t} \Biggl \lfloor \frac{m}{\prod\limits_{j=1}^{i} r_j} \Biggr \rfloor \] 
\end{lemma}	
\begin{IEEEproof}
For $i \in \left[ t \right]$, we have
\begin{align} \label{intermi}
\te_i \prod\limits_{j=1}^{i} r_j &=  \sum\limits_{j=1}^{i+1} \prod\limits_{l=j}^{i} r_l.
\end{align}
Next, observe that 
\[ \sum\limits_{i=0}^{t} \Biggl \lfloor \frac{m}{\prod\limits_{j=1}^{i} r_j} \Biggr \rfloor =  \sum\limits_{i=0}^{t} \left( \beta r_t + \alpha_t \right) \prod\limits_{j=i+1}^{t} r_j + \sum\limits_{j=i}^{t-1} \alpha_{j} \prod\limits_{l=i+1}^{j} r_l. \]
Also, from \eqref{intermi} and the above representation of $m$, 
\[  \Biggl \lfloor \frac{m}{\prod\limits_{i=1}^{t} r_i } \Biggr \rfloor e_t \prod_{i=1}^{t} r_i  = \left( \beta r_t + \alpha_t \right) \sum\limits_{i=0}^{t-1} \prod\limits_{j=i+1}^{t} r_j. \] 
Therefore, it suffices to prove that 
\[ \sum\limits_{i=0}^{t-1} \alpha_i \sum\limits_{j=1}^{i+1} \prod\limits_{l=j}^{i} r_l  =\sum\limits_{i=0}^{t} \sum\limits_{j=i}^{t-1} \alpha_{j} \prod\limits_{l=i+1}^{j} r_l. \]
The above equation can be verified to be true by comparison of coefficients of $\alpha_i$, $i \in \left[ t-1 \right]$, thus completing the proof.
\end{IEEEproof}

\begin{theorem}\label{unifirreg}
If an $\left[ n,k,d \right]$ code $\mathcal{C}$ has $\left( \bm{r} , t\right) $ all-symbol locality, then \[ d \leq n- k+1 - \sum\limits_{i=1}^{t} \left \lfloor \frac{k-1}{\prod_{j=1}^{i} r_j} \right \rfloor. \]   
\end{theorem}
\begin{IEEEproof} The proof is similar to that of \cite{tamo2016bounds}, and only key ideas are presented. We obtain a $k-1$ sized subset of vertices $S$ and prove a lower bound on $| \text{Cl} \left(S \right)|$ by applying Lemma 1 repeatedly. Consider the recovering graph $G$ of $\C$. From Lemma 1, there exists a set of vertices $S_0$
 such that its expansion ratio is atleast $e_t$. Let the induced subgraph on $V \setminus \text{Cl} \left(S_0 \right) $ be $G_1$, which is an $(\bm{r},t)$ recovering graph. Apply Lemma 1 on $G_1$ and continue this process until the number of remaining vertices in the graph $G_l$ after $l$ steps is lesser than $\prod_{j=1}^{t} r_j$. 

Now, continue by viewing the graph $G_l$ as an $([r_1,\ldots,r_{t-1}],t-1)$ recovering graph. By Lemma 1, there exists a set of vertices $S_l$ with $|S_l| \leq \prod_{j=1}^{t-1} r_j$ and expansion ratio at least $\te_{t-1}$ (see Lemma 2 for definition). Continue the coloring process going through $([r_{1},\ldots,r_i],i)$ recovering graphs containing sets of vertices of size at most $\prod_{j=1}^{i} r_j$ and expansion ratio at least $\te_i$ for $i=t-2,\ldots,1$ till $k-1$ vertices are colored.  

By keeping track of the expansion ratios and the number of applications of Lemma 1, we get the following lower bound on $ | \text{Cl} \left( S \right) |$:
 \[ | \text{Cl} \left( S \right) |  \geq \Biggl \lfloor \frac{k-1}{\prod\limits_{i=1}^{t} r_i } \Biggr \rfloor e_t \prod_{i=1}^{t} r_i  + \sum\limits_{i=0}^{t-1} \alpha_i \te_i \prod_{j=1}^{i} r_j \]
 where $k-1 =  \sum_i \left( \alpha_i \prod\limits_{j=1}^{i} r_j\right)$.
 Using Lemma 2, 
  \[ | \text{Cl} \left( S \right) |  \geq  \sum\limits_{i=0}^{t} \Biggl \lfloor \frac{k-1}{\prod\limits_{j=1}^{i} r_i} \Biggr \rfloor. \]
  Since $\rk(\text{Cl} \left( S \right))=\rk(S)<k$, $d \leq n - | \text{Cl} \left( S \right) |$, which results in the bound of the theorem.
\end{IEEEproof}

\section{Unequal Locality with Availability}
We now consider the case of unequal locality where different coordinates have possibly different, but regular recovery profiles with availability $t$. That is, the $i$-th coordinate has a length-$t$ recovery profile of the form $[r_i\ r_i\ \cdots\ r_i]$. We will consider information locality for this case and prove an upper bound on minimum distance.
\begin{definition}
An $\left[ n,k,d \right]$ code $\mathcal{C}$ has information locality profile $\{k_1,k_2,\ldots,k_r\}$ with availability $t$ if $k_i$ is the number of information coordinates with locality $i$ and availability $t$. 
\end{definition}
A modified version of Algorithm 1, which we refer to as Algorithm 2, is used in the proof. Algorithm 2 is identical to Algorithm 1 except for Step 3, which becomes
$$3:\text{Set }i=i+1,\text{ Choose }j\in [n]\setminus S_{i-1}\text{ with minimal locality}$$
\begin{theorem}
If $\mathcal{C}$ is an $ \left[ n,k,d \right]$ linear code with information locality profile  $ \left \lbrace k_{1}, k_{2}, \ldots, k_{r} \right \rbrace $ with availability $t$, then 
$$ d \leq n - k +2 - t \left( \sum\limits_{j=1}^{r-1} \left \lceil {\frac{k_{j}}{t\left( j-1 \right) +1} }\right \rceil \right) - \left \lceil {\frac{ t \left( k_{r}-1 \right) +1 }{t \left( r-1 \right) +1}} \right \rceil. $$ 
\end{theorem}
\begin{IEEEproof}
We use Algorithm 2 with the code $\C$. Let $l$ be the number of iterations of Algorithm 2. Consider  two cases depending on how Algorithm 2 terminates.

\noindent\emph{Case 1}: $S_{l}$ is formed in line $5$. 

As in the proof of Theorem \ref{th1}, we get 
\begin{equation} \label{Sbound2}
|S| \geq k-1 + tl.
\end{equation} 
Let $l_j$ be the number of iterations in which coordinates with locality $j$ are chosen. In these $l_j$ iterations, the rank of $S$ increases by $k_j$ for $0 \leq j \leq r-1$. Since $\text{rank} \left( S_{l} \right) = k-1$ and coordinates with least locality are preferred in Step 3, for $j=r$, the rank increment for the $l_{r}$ iterations is $k_{r} -1 $. Now, as in the proof of Theorem \ref{th1},
\begin{align*}
k_j &\leq l_j \left( 1+t \left( j-1 \right) \right) \hspace{5mm} \forall j \in \left[ r-1 \right], \\
k_{r} -1 &\leq l_{r} \left( 1 + t \left( r-1 \right) \right).
\end{align*}  
Since $l = \sum_{j=1}^{r} l_{j}$, we have 
\begin{equation} \label{lbound2}
l \geq \sum\limits_{j=1}^{r-1} \left \lceil {\frac{k_j}{t \left( j-1 \right) +1}} \right \rceil + \left \lceil {\frac{k_r - 1}{t \left( r-1 \right) +1 }} \right \rceil.
\end{equation}
Plugging \eqref{lbound2} in \eqref{Sbound2}, we get
\begin{align*}
|S| &\geq k-1 + t \left( \sum\limits_{j=1}^{r-1} \left \lceil {\frac{k_j}{t \left( j-1 \right) +1}} \right \rceil + \left \lceil {\frac{k_r - 1}{t \left( r-1 \right) +1 }} \right \rceil \right) \nonumber \\
&\geq k-2 + t \left( \sum\limits_{j=1}^{r-1} \left \lceil {\frac{k_j}{t \left( j-1 \right) +1}} \right \rceil \right) + \left \lceil {\frac{t \left( k_r - 1\right) +1}{t \left( r-1 \right) +1 }} \right \rceil,  
\end{align*}
where the manipulations for the last step are same as before. Using $d\le n-|S|$, the distance bound follows.

\paragraph*{Case 2} $S_{l}$ is formed in line $9$. 

Note that \eqref{expS2} holds in this case. Since $\text{rank} \left( S_{l-1} \cup \Gamma_{a+1} \left( c_{l} \right) \right) = k$, in the last $l_r$ iterations, the rank increment is now $k_r$ instead of $k_r-1$ in Case 1 above. Lower bounds on $l_j$, $1 \leq j \leq r-1$ are the same as in Case 1. For $l_r$, we get a lower bound from the following inequality
\[ k_r \leq \left( l_r -1 \right) \left( 1+ t \left(r-1 \right) \right) + 1 + \left( a+1 \right) \left( r-1 \right). \]
Adding the lower bounds for $l_j$,
\begin{equation} \label{lbound2case2}
l-1 \geq \sum\limits_{j=1}^{r-1} \left \lceil {\frac{k_j}{t \left( j-1 \right) +1}} \right \rceil + \left \lceil {\frac{k_r-ar-r+a }{t \left( r-1 \right) +1 }} \right \rceil.	 
\end{equation}
Plugging \eqref{lbound2case2} in \eqref{expS2}, we have 
\begin{align} \label{Sbound2case2}
|S| &\geq k-1 + a+\nonumber \\ 
&\quad t \left( \sum\limits_{j=1}^{r-1} \left \lceil {\frac{k_j}{t \left( j-1 \right) +1}} \right \rceil + \left \lceil {\frac{k_r-ar-r+a}{t \left( r-1 \right) +1 }} \right \rceil \right). 
\end{align}
Let $\Omega=1+i\left(r-1 \right)$. Using $t \left \lceil x \right \rceil \geq \left \lceil tx \right \rceil $, 
\begin{align}
 t  \left \lceil {\frac{k_r-\Omega_{a+1}}{\Omega_{t} }} \right \rceil &\geq  \left \lceil {\frac{t \left(k_r-\Omega_{a+1} \right)}{\Omega_{t} }} \right \rceil \label{eq:19} \\
 &= \left \lceil {\frac{t \left( k_r - 1 \right)+1 -t \left( \Omega_{a+1}-1\right)-1}{\Omega_t}} \right \rceil \nonumber \\
 &= \left \lceil {\frac{t \left( k_r - 1 \right)+1}{\Omega_t} - \left(a+1 \right) + \frac{a}{\Omega_t}} \right \rceil. \label{eq:20} 
\end{align}
Substituting \eqref{eq:20} in \eqref{Sbound2case2}, and using $d \leq n - |S|$, we get the desired bound. 
\end{IEEEproof}

\section{Optimal LRCs with Availability}
In the previous sections, we had derived minimum distance upper bounds for locally recoverable codes (LRCs) with availability under different scenarios. In this section, we will consider examples of constructions of LRCs with availability and compare  their minimum distances with the derived distance upper bounds. In some cases, we obtain optimal constructions where the minimum distance meets the upper bound.

 \subsection{Regular recovery and locality with availability}
In this section, we revert to the notion of equal recovery and locality, and consider LRCs having $\left(r,t \right)$ locality with availability. 
\subsubsection{All-symbol locality}
The upper bound on minimum distance for LRCs with availability $t$ and all-symbol locality $r$ is given by \eqref{tbbound}. The tightness of this bound for arbitrary $t$ has not been fully settled. We consider a generalization of the polynomial-evaluation construction of LRC codes in Example 6 of \cite{tamo2014family} and show optimality for some specific cases by computational methods. Later, we prove optimality for the case of $r=k-1$ and arbitrary $t$. 

For the sake of clarity and completeness, we briefly outline Example 6 of \cite{tamo2014family} below. 
\begin{example}\label{example_text}
An $\left( n=16,k,r=3,t=2 \right)$ LRC is constructed over $\mathbb{F}_{16}$ by generating orthogonal partitions from cosets of two copies of $\mathbb{F}^{+}_{4}$ denoted $H_{1}= \left \lbrace 0,1,\alpha,\alpha^{4} \right \rbrace $, $H_{2} = \left \lbrace 0,\alpha^{2},\alpha^{3},\alpha^{6} \right \rbrace $, where $\alpha$ is the residue class of $x$ modulo $x^{4} + x + 1$. The annihilator polynomials of $H_{1}$ and $H_{2}$, denoted $g_1$ and $g_2$, respectively, are
\begin{align*}
g_{1}(x) = x^{4} + \alpha^{10}x^{2} + \alpha^{5}x, \\
g_{2}(x) = x^{4} + \alpha^{14}x^{2} + \alpha^{11}x.
\end{align*} 
The orthogonal partitions that are generated by $H_1$, $H_2$ and their cosets are 
\begin{align*}
 	\mathcal{A}_{1} = \left \lbrace \left \lbrace 0,1,\alpha,\alpha^{4} \right \rbrace , \left \lbrace \alpha^{2},\alpha^{8},\alpha^{5},\alpha^{10} \right \rbrace , \left \lbrace \alpha^{3},\alpha^{14},\alpha^{9},\alpha^{7} \right \rbrace, \right. \\
\left \lbrace \alpha^{6},\alpha^{13},\alpha^{11},\alpha^{12} \right \rbrace  \rbrace, \\
\mathcal{A}_{2} = \left \lbrace \left \lbrace 0,\alpha^{2},\alpha^{3},\alpha^{6} \right \rbrace ,\left \lbrace 1,\alpha^{8},\alpha^{14},\alpha^{13} \right \rbrace  ,\left \lbrace \alpha,\alpha^{5},\alpha^{9},\alpha^{11} \right \rbrace , \right. \\\left \lbrace \alpha^{4},\alpha^{10},\alpha^{7},\alpha^{12} \right \rbrace \rbrace.
\end{align*}
%The key observation that we use in proving optimality is the absence of the $x^{3}$ term in both the annihilator polynomials.
The basis of $\mathcal{F}_{\mathcal{A}_{1}}^{3} \bigcap \mathcal{F}_{\mathcal{A}_{2}}^{3}$ is obtained by choosing polynomials of distinct degrees that can be expressed as a linear combination of the basis of both $\mathcal{F}_{\mathcal{A}_{1}}^{3}$ and $\mathcal{F}_{\mathcal{A}_{2}}^{3}$. We find that the basis of $\mathcal{F}_{\mathcal{A}_{1}}^{3} \bigcap \mathcal{F}_{\mathcal{A}_{2}}^{3}$ comprises of  polynomials of degrees $0,1,2,4,6,8,9,10,12$. Table \ref{table_example} summarizes the possible dimensions along with two distance lower bounds: (a) $n- \max_{f_a \in V_{m}} \text{deg} \left( f_{a} \right)$, (b) $n - \max\limits_{f_a \in V_{m}} \text{deg} \left( \gcd \left( f_a,x^{16}-x \right) \right)$. The second lower bound is evaluated computationally. The distance upper bound in \eqref{tbbound} is also shown. The second lower bound is tighter than the first for $k=6,7$. For $k=7$, the second lower bound meets the upper bound, thus giving an optimal code. For $k=8,9$, our computations for the second lower bound did not terminate. 
\begin{table}[htb]
\renewcommand{\arraystretch}{1.5}
\setlength{\tabcolsep}{20pt}
\caption{Lower and upper bound on minimum distance for Example \ref{example_text}. }
\label{table_example}
    \centering
    $\begin{array}{| C | C | C | C |}
        \hline
        $k$ & LB 1 & LB 2 & UB \\
        \hline
         4 & 12 &12& 12\\ \hline
         5 & 10 &10& 11 \\ \hline
         6 & 8 &9& 10 \\ \hline
         7 & 7 &8& 8 \\ \hline
         8 & 6 &$-$& 7 \\ \hline
         9 & 4 &$-$& 6 \\ \hline         
    \end{array}$
\end{table}
\end{example}
We generalize the construction in the above example to arbitrary $t$ and show that it is optimal for $k=r+1$. 	
\begin{construction}\label{const}
Let $r+1 = p^{l}$, $p$: prime, $l \geq 1$, $t \geq 2$. Let $n = (r+1)^{t},k=r+1$ and $A = \mathbb{F}_{(r+1)^{t}}$. The additive subgroup of $\mathbb{F}_{(r+1)^{t}}$ can be written as
$$ \mathbb{F}^{+}_{(r+1)^{t}} \cong \{[a_1,\ldots,a_t]:a_i\in\mathbb{F}^{+}_{(r+1)}\}. $$ 
Consider $t$ subgroups $H_i=\{[0,\ldots,0,a_i,0,\ldots,0]:a_i\in\mathbb{F}^{+}_{(r+1)}\}$ of $ \mathbb{F}^{+}_{(r+1)^{t}}$ of size $r+1$ for $i\in[t]$. Let $g_{i}$ be the annihilator polynomial of $H_i$. Let $\mathcal{A}_{i}$ be the partitions of $A$ induced by the cosets of $H_i$. We have
$$ \mathbb{F}_{\mathcal{A}_{i}}[x] = \langle 1,g_{i}(x), g_{i}(x)^{2} \ldots g_{i}(x)^{\frac{n}{r+1} -1} \rangle.$$
 Since $\bigcap\limits_{i=1}^{t} H_{i} = \left \lbrace 0 \right \rbrace $, $ \left \lbrace \mathcal{A}_{i} \right \rbrace $ are orthogonal partitions.

Now, a crucial observation is the following. Since $H_{i}$ is a copy of $\mathbb{F}_{r+1}^{+}$, we have $ \sum\limits_{h \in H_{i}} h = 0$. 
It follows that the coefficient of $x^{r}$ in $g_{i}(x)=\prod_{h\in H_{i}}(x-h)$ is $0$ $\forall i \in [t]$. Therefore, $g_{i}(x)$ is of degree $r+1$ and is contained in  $\bigcap\limits_{i=1}^{t} \mathcal{F}_{\mathcal{A}_{i}}^{r}$ as degrees $1$ to $r-1$ are contained in each $\mathcal{F}_{\mathcal{A}_{i}}^{r}$. Using this, we see that
$$V_{r+1} = \bigcap_{i=1}^{t} \mathcal{F}_{\mathcal{A}_{i}}^{r} \bigcap P_{r+1} = \langle 1,x \ldots x^{r-1}, g_1(x) \rangle.$$
To encode a message $a \in \mathbb{F}_{(r+1)^{t}}^{r+1}$, we define the encoding polynomial 
$$f_{a}(x) = \sum_{i=0}^{r-1} a_{i}x^{i} + a_{r+1}g_1(x).$$
The code is obtained by evaluating $f_{a}$ on the points of $A$.
\end{construction}

\begin{theorem}
The $\left( (r+1)^{t}, r+1, r,t \right)$ LRC code from Construction \ref{const} is optimal.
\end{theorem}
\begin{proof}
 % . Since, %$$d \geq n - \max_{f_{a} \in V_{r+1},  a \in \mathbb{F}_{(r+1)^{t}}^{r+1}} \text{deg}f_{a}$$
 Since $r+1$ is the maximum degree of the encoding polynomials, we have $d \geq n -  (r+1)$. By the bound \eqref{tbbound},
$$d \leq n - \left(k-1 + \sum\limits_{i=1}^{t} \left \lfloor{\frac{k-1}{r^{i}}}\right \rfloor \right) = n - (r+1),$$ 
and the proof is complete.
\end{proof}
% For $r<k-1$, a basis is chosen in the same way as mentioned in the above example. More formally, to obtain the basis of the intersection we define a matrix $M_l$ corresponding to the $l^{\text{th}}$ recovering set, $l \in \left[ t \right]$ where $M_l \left( i,j \right)$ is the coefficient of $x^{i}$ of the $j^{\text{th}}$ basis (ordered in increasing order of their degrees). We say that $M_l$ is 'truncated to $m$' if the degree of the last basis is atmost $m$. Consequently, the rows will also be truncated till degree $m$. Let $M \left( m \right)$ denote the matrix obtained by augmenting $M_l$, $l \in \left[ t \right]$, such that each $M_l$ is truncated to $m$. It is not hard to see that the nullspace of $M \left( m \right)$ is isomorphic to $\bigcap_{i=1}^{t} \mathcal{F}_{\mathcal{A}_{i}}^{r} \bigcap P_{m}$, thus providing a method to obtain the dimension of the code by getting the rank of $M \left( m \right)$. The distance lower bound can be obtained by identifying the maximum degree of the polynomials in the nullspace of $M \left( m \right)$.
\subsubsection{Information locality}
An upper bound on minimum distance for LRCs with information locality $r$ and availability $t$ is given by \eqref{infobd}. We extend the parity-check matrix construction of \cite{DBLP:journals/corr/HaoX16} to include availability, and show that it achieves the distance upper bound for $n \geq k \left( tr+1 \right)$. Let 
\begin{equation}
  \label{eq:8}
  \Gamma = n-k+1 - \left \lceil \frac{t\left( k-1 \right) +1 }{t\left( r-1 \right) +1 } \right \rceil.
\end{equation}
If any $\Gamma$ columns of a parity-check matrix are linearly independent, the corresponding code meets the minimum distance upper bound in \eqref{infobd}.
\begin{construction}
Let $tr+1 | n$ and $tr+1 \nmid \Gamma$. Let $v = n/(tr+1)$, $u = n-k-vt$. 

\noindent\emph{Local parity checks}: Define the $t\times(1+tr)$ matrix
\[ 
\sbox1{$\begin{matrix}
1 \\ 1 \\ \vdots \\ 1
\end{matrix}$}
H_1 = \left[ \begin{array}{c|c}
\usebox{1} & I_t \otimes \underbrace{\left( 1 1 \dots 1 \right)}_\text{r} \end{array} \right],
 \] 
where $\otimes$ denotes matrix tensor product. The local parity checks are constructed as
a $vt\times n$ matrix
$$H_{\text{local}}=\begin{bmatrix}
H_1    &0        &\cdots &0\\
0       &H_1     &\cdots &0\\
\vdots&\vdots&\ddots&\vdots\\
0        &0       &\cdots&H_1
\end{bmatrix}.$$
The $v$ coordinates numbered 1, $2+tr$, $3+2tr$, $\ldots$, $v+(v-1)tr$ have locality $r$ with availability $t$, and are referred to as \emph{availability} columns. The nonzero coordinates in every row of $H_{\text{local}}$ is called a repair group.

\noindent\emph{Global parity checks}: For a vector $\bm{\a}=[\a_1\ \a_2\ \cdots\ \a_r]$ with $\a_i\in\F_{q^m}$, define the $u\times r$ matrix $M(\bm{\a})$ as follows:
$$M(\bm{\a})=\begin{bmatrix*}[l]
\a_1&\a_2&\cdots&\a_r\\
\a^q_1&\a^q_2&\cdots&\a^q_r\\
\vdots&\vdots&\vdots&\vdots\\
\a^{q^{u-1}}_1&\a^{q^{u-1}}_2&\cdots&\a^{q^{u-1}}_r
\end{bmatrix*}.$$
The $b$-th column of $M(\bm{\a})$ is denoted $c(\a_b)$. Let $\a_{i,j,h}\in\F_{q^m}$ for $1\le i\le v$, $1\le j\le t$, $0\le h\le r$, and define the vector $\bm{\a}_{i,j}=[\a_{i,j,1}\ \a_{i,j,2}\ \cdots\ \a_{i,j,r}]$. Define the $u\times tr$ matrix $M_i=[M(\bm{\a}_{i,1})\ M(\bm{\a}_{i,2})\ \cdots\ M(\bm{\a}_{i,t})]$. The global parity checks are constructed as an $u\times n$ matrix
\begin{align*}
H_{\text{global}}=[c(\a_{1,1,0})\,M_1\ \ c(\a_{2,1,0})\,M_2\ \cdots\ c(\a_{v,1,0})\,M_v].
\end{align*}
Further, we require that $m \geq v(t( r-1) +1)$ and that 
$$\left\{\alpha_{i,1,0}-\sum\limits_{l=1}^{t} \alpha_{i,l,r},\ \alpha_{i,j,h}-\alpha_{i,j,r}\right\},$$ 
 for $\ 1\le i\le v$, $1\le j\le t$, $1\le h\le r-1$, are linearly independent over $\mathbb{F}_q$. Finally, the overall parity-check matrix is defined as $H=\begin{bmatrix*}[l]H_{\text{local}}\\H_{\text{global}}\end{bmatrix*}$. 
\end{construction} 
 \begin{theorem}
 For $n \ge k \left(tr+1 \right)$, the linear code obtained using Construction 2 is a $q^m$-ary $(n,k)$ LRC with information locality $r$ and availability $t$. The code meets the distance upper bound \eqref{infobd}.
 \end{theorem}
 \begin{proof} 
Since $k\le v$, we choose the $k$ information symbols among the availability columns, and ensure that the code has information locality $r$ with availability $t$. The rank condition for the parity part of the matrix $H$ will be proved later.

Choose $\Gamma$ columns arbitrarily from $H$. Let $\Delta$ be the number of non-zero rows among the chosen $\Gamma$ columns. We will first show that $\Delta\ge\Gamma$. 

The $u$ rows of $H_{\text{global}}$ are nonzero. Let $x\le v$ be the number of availability columns chosen among the $\Gamma$ columns. This results in $xt$ distinct nonzero rows in $H_{\text{local}}$. Let $x_0=\Gamma/(tr+1)$. Additional number of rows among the chosen non-availability columns depends on the following two cases.  

\noindent\textbf{Case 1}: $x\ge \lceil x_0 \rceil$.

\noindent In this case, there may be no rows obtained from non-availability columns, and we have the lower bound
\begin{equation} \label{orig}
\Delta\ge \Delta_{\text{LB}}\triangleq u + xt\ge u+t \lceil x_0\rceil. 
\end{equation}
To relate $\Delta$ to $\Gamma$, we start with the following observation:
$$\frac{t(k-1)}{t(r-1)+1}<\left \lceil \frac{t(k-1) +1 }{t( r-1) +1 } \right \rceil\le \frac{kt}{t(r-1)+1}+1.$$
Using the above in \eqref{eq:8} and simplifying, we obtain the following bounds on $\Gamma$:
\begin{align}
  \label{eq:9}
  u+tx_0\le \Gamma< u+tx_0+1.
\end{align}
Therefore, if $t \lceil x_0 \rceil \geq tx_0+1$, we clearly have $\Delta_{\text{LB}}\ge \Gamma$. On the other hand, if $t\lceil x_0 \rceil < tx_0+1$, we have
\begin{equation}
  \label{eq:10}
  u+tx_0\le \Delta_{\text{LB}}< u+tx_0+1.
\end{equation}
So, from \eqref{eq:9} and \eqref{eq:10}, both $\Delta_{\text{LB}}$ and $\Gamma$ are equal to $u+\lceil tx_0\rceil$, which is the unique integer lying in the interval $[u+tx_0,u+tx_0+1)$. So, we have $\Delta_{\text{LB}}=\Gamma$.

\noindent\textbf{Case 2}: $0 \leq x \leq \lfloor x_0 \rfloor$.

\noindent In this case, at least $\left\lceil\dfrac{\Gamma-x(1+tr)}{r}\right\rceil$ additional nonzero rows will be included among the nonavailability columns. So, we get the following lower bound on $\Delta$:
\begin{align}
  \Delta&\ge u+xt+\left\lceil\frac{\Gamma-x(1+tr)}{r}\right\rceil\nonumber\\
&\overset{(a)}{=}u+\left\lceil\frac{\Gamma-x}{r}\right\rceil\overset{(b)}{\ge}u+\left\lceil\frac{\Gamma-x_0}{r}\right\rceil\nonumber\\
&\overset{(c)}=u+\lceil tx_0 \rceil\overset{(d)}{=}\Gamma,
  \label{eq:11}
\end{align}
where $(a)$ was obtained by moving the integer $xt$ inside the ceil, $(b)$ results because $x\le x_0$, $(c)$ is obtained by plugging in the expression for $x_0$, and $(d)$ results from \eqref{eq:9}. This concludes the proof that $\Delta\ge\Gamma$.

The rest of the proof is to show that the submatrix $H_s$ made of the chosen $\Gamma$ columns has full column rank. This is similar to the proof in \cite{DBLP:journals/corr/HaoX16}, and we provide a brief outline pointing out the main differences to account for the availability columns. Consider a repair group $\{i_0,i_1,\ldots,i_r\}$, where $i_0$ is the availability column. Let $S=\{j:i_j\text{ is chosen}\}$ be the set of chosen columns in the repair group. If $|S|>1$, for every $j\in S$, $j<\max(S)$, set column $i_j$ as the difference of column $i_j$ and column $i_{\max(S)}$. Let $H'_s$ denote the modified $H_s$ after the above operations on all repair groups.

We now reduce $H'_s$ to a square matrix by deleting rows. First consider the locality part of $H'_s$. Delete the all-zero rows. An availability column will have a 1 if no non-availability column is chosen in any of its $t$ repair groups. If there are multiple 1s in an availability column, retain only one such row and delete the others. Finally, in the global part of $H'_s$, delete rows from the bottom to obtain a square matrix and denote it $H''_s$. Now, in the locality part of $H''_s$, every row has a single 1, and every column is either all-zero or has a single 1. So, we get the following structure for a suitably column-permuted $H''_s$:
\begin{equation}
  \label{eq:12}
  \left[\begin{array}{c|c}
\text{Identity}&\text{All-zero}\\
\hline
H_1&H_2
\end{array}\right],
\end{equation}
where $H_1$ and $H_2$ are the global part. Denoting the first row of the square matrix $H_2$ by $\bm{\g}=[\g_1,\g_2,\ldots,\g_l]$, we see that $H_2=M(\bm{\g})$ with $u=l$, and that the $\g_i$ are from the set $\left \lbrace \alpha_{i,1,0}-\sum\limits_{l=1}^{t} \alpha_{i,l,r_l}, \alpha_{i,j,h}-\alpha_{i,j,r_j} \mid h \in [r_j-1] \right \rbrace $, where $i$ and $j$ depend on the arbitrary chosen $\Gamma$ columns. By our construction and a proof similar to that of Lemma 4 of \cite{DBLP:journals/corr/HaoX16}, the elements of $\bm{\g}$ are  linearly independent over $\mathbb{F}_q$. As a result, the determinant of $H_2$ is nonzero implying that the originally chosen arbitrary $\Gamma$ columns are linearly independent.

The proof for the rank of the parity part of $H$ is similar to the above and we skip the details.
 \end{proof}	
 %Note that, there exist $m$ such that
% \[ n-k \geq \frac{ntr}{tr+1} \geq m \geq n \frac{t \left( r-1 \right) +1}{tr+1}. \] 
%Therefore, $q^{n-k} \ge q^m$. 
%For $q \leq {{n}\choose{k+\mu}}^{\frac{1}{m}}$ where $\mu = \left \lceil \frac{t(k-1)+1}{t(r-1)+1} \right \rceil, m = n\frac{t(r-1)+1}{tr+1}$, our construction results in a smaller field size when compared to \cite{wang2014repair}. 
We note that Construction 2 gives an explicit construction unlike  \cite{wang2014repair}. 
 Finally, we remark that Construction 2 can easily be extended to construct codes with unequal information locality and availability as defined in Section~\ref{section:3}.     
 %\subsubsection{Uniform Locality with Irregular Recovering Sets}
 
\subsection{Irregular recovery with availability}
We consider two examples of LRCs having irregular recovery with availability. The first example is the same as Example 5 of \cite{tamo2014family}. The second example is a similar construction applied to a larger field. We compare their distances with the upper bound of Theorem \ref{unifirreg}. We show the second example is optimal by showing that the distance satisfies the upper bound of Theorem \ref{unifirreg} with equality.
\begin{example}
An $\left(n=12,k=4,r_1=3,r_2=2,t=2 \right)$ LRC is constructed over $\mathbb{F}_{13}$ by generating orthogonal partitions, $\mathcal{A}$ and $\mathcal{A}'$, from the cosets of the multiplicative subgroups generated by $5$ and $3$, respectively. The partitions are as follows:
\[ \mathcal{A} = \left \lbrace \left \lbrace 1,5,12,8 \right \rbrace , \left \lbrace 2,10,11,3 \right \rbrace \left \lbrace 4,7,9,6 \right \rbrace \right \rbrace, \]
 \[ \mathcal{A}' = \left \lbrace \left \lbrace 1,3,9 \right \rbrace , \left \lbrace 2,6,5 \right \rbrace , \left \lbrace 4,12,10 \right \rbrace , \left \lbrace 7,8,11 \right \rbrace \right \rbrace. \]
Since the constant polynomials with respect to $\mathcal{A}$ and $\mathcal{A}'$ are $x^4$ and $x^3$, respectively, (see Section \ref{sec:polyn-eval-constr}), we have 
\[ \mathbb{F}_{\mathcal{A}} \left[ x \right] = \langle 1,x^4,x^8 \rangle, \quad \mathbb{F}_{\mathcal{A}'} \left[ x \right] = \langle 1,x^3,x^6,x^9 \rangle. \]
Finding the basis of $\mathcal{F}_{\mathcal{A}}^{3} \bigcap \mathcal{F}_{\mathcal{A}'}^{2}$ and truncating it to obtain a $k=4$ dimensional subspace, we have
\[ V_6 = \langle 1, x, x^4, x^6 \rangle. \]
Since $\max\limits_{f_a \in V_{6}} \text{deg} \left( f_{a} \right)$ is $6$, this code has distance $d \geq 6$. Evaluating the upper bound on distance from Theorem \ref{unifirreg}, we have $d \leq 8$, which differs from the upper bound.
\end{example} 

\begin{example} 
An $\left(n=32,k=8,r_1=7,r_2=3,t=2 \right)$ LRC is constructed over $\mathbb{F}_{32}$ by generating orthogonal partitions, $\mathcal{A}$ and $\mathcal{A}'$, from cosets of copies of $\mathbb{F}_{8}^{+}$ and $\mathbb{F}_{4}^{+}$ denoted $H$ and $H'$, respectively. Let $\alpha\in\mathbb{F}_{32}$ be primitive satisfying $\alpha^5+\alpha^2+1=0$. We have
\[ H = \left( 0,1,\alpha,\alpha^2,\alpha^5,\alpha^{11},\alpha^{18},\alpha^{19} \right), \quad H' = \left( 0,\alpha^3,\alpha^4,\alpha^{21} \right). \]
$\mathcal{F}_{\mathcal{A}}^{7}$, $\mathcal{F}_{\mathcal{A}'}^{3}$  and annihilator polynomials of $H$, $H'$ are obtained as defined in Section \ref{sec:polyn-eval-constr}. By linear algebraic techniques, we can find the dimension and basis of  $V_m=\mathcal{F}_{\mathcal{A}}^{7} \bigcap \mathcal{F}_{\mathcal{A}'}^{3}\cup P_{m}$ numerically. It can be verified numerically that for $m=9$, we have the dimension $k=8$. Thus, evaluating the distance upper bound of Theorem \ref{unifirreg}, we get $d \leq 23$. Since $\max\limits_{f_a \in V_{9}} \text{deg} \left( f_{a} \right)$ is $9$, we have $d \geq 23$. Thus, the bound of Theorem \ref{unifirreg} is met with equality. 
\end{example} 

\section{Conclusion}
We derived new upper bounds on minimum distance for codes with unequal all-symbol locality and availability. We presented a generalization of a construction for LRCs with availability, that attains the upper bound on minimum distance for arbitrary $t$ and $r=k-1$. An explicit parity-check matrix construction that meets the upper bound on minimum distance for LRCs with information locality and availability was also obtained for $n \geq k \left(tr+1 \right)$. Future work includes finding optimal constructions for LRCs with availability for higher values of $k$ (or lower values of $r$) and finding constructions that meet the bounds proposed in this paper for a larger range of parameter values. 

\bibliographystyle{IEEEtran}
\bibliography{IEEEabrv,ref}

\begin{thebibliography}{10}
\providecommand{\url}[1]{#1}
\csname url@samestyle\endcsname
\providecommand{\newblock}{\relax}
\providecommand{\bibinfo}[2]{#2}
\providecommand{\BIBentrySTDinterwordspacing}{\spaceskip=0pt\relax}
\providecommand{\BIBentryALTinterwordstretchfactor}{4}
\providecommand{\BIBentryALTinterwordspacing}{\spaceskip=\fontdimen2\font plus
\BIBentryALTinterwordstretchfactor\fontdimen3\font minus
  \fontdimen4\font\relax}
\providecommand{\BIBforeignlanguage}[2]{{%
\expandafter\ifx\csname l@#1\endcsname\relax
\typeout{** WARNING: IEEEtran.bst: No hyphenation pattern has been}%
\typeout{** loaded for the language `#1'. Using the pattern for}%
\typeout{** the default language instead.}%
\else
\language=\csname l@#1\endcsname
\fi
#2}}
\providecommand{\BIBdecl}{\relax}
\BIBdecl

\bibitem{gopalan2012locality}
P.~Gopalan, C.~Huang, H.~Simitci, and S.~Yekhanin, ``On the locality of
  codeword symbols,'' \emph{IEEE Transactions on Information Theory}, vol.~58,
  no.~11, pp. 6925--6934, 2012.

\bibitem{6620541}
N.~Silberstein, A.~S. Rawat, O.~O. Koyluoglu, and S.~Vishwanath, ``Optimal
  locally repairable codes via rank-metric codes,'' in \emph{Information Theory
  Proceedings (ISIT), 2013 IEEE International Symposium on}, July 2013, pp.
  1819--1823.

\bibitem{6620540}
I.~Tamo, D.~S. Papailiopoulos, and A.~G. Dimakis, ``Optimal locally repairable
  codes and connections to matroid theory,'' in \emph{Information Theory
  Proceedings (ISIT), 2013 IEEE International Symposium on}, July 2013, pp.
  1814--1818.

\bibitem{DBLP:journals/corr/HaoX16}
\BIBentryALTinterwordspacing
J.~Hao and S.~Xia, ``Bounds and constructions of locally repairable codes:
  Parity-check matrix approach,'' \emph{CoRR}, vol. abs/1601.05595, 2016.
  [Online]. Available: \url{http://arxiv.org/abs/1601.05595}
\BIBentrySTDinterwordspacing

\bibitem{tamo2014family}
I.~Tamo and A.~Barg, ``A family of optimal locally recoverable codes,''
  \emph{IEEE Transactions on Information Theory}, vol.~60, no.~8, pp.
  4661--4676, 2014.

\bibitem{6620355}
L.~Pamies-Juarez, H.~D.~L. Hollmann, and F.~Oggier, ``Locally repairable codes
  with multiple repair alternatives,'' in \emph{Information Theory Proceedings
  (ISIT), 2013 IEEE International Symposium on}, July 2013, pp. 892--896.

\bibitem{rawat2014locality}
A.~S. Rawat, D.~S. Papailiopoulos, A.~G. Dimakis, and S.~Vishwanath, ``Locality
  and availability in distributed storage,'' in \emph{2014 IEEE International
  Symposium on Information Theory}.\hskip 1em plus 0.5em minus 0.4em\relax
  IEEE, 2014, pp. 681--685.

\bibitem{wang2014repair}
A.~Wang and Z.~Zhang, ``Repair locality with multiple erasure tolerance,''
  \emph{IEEE Transactions on Information Theory}, vol.~60, no.~11, pp.
  6979--6987, 2014.

\bibitem{WangZL15}
A.~Wang, Z.~Zhang, and M.~Liu, ``Achieving arbitrary locality and availability
  in binary codes,'' in \emph{2015 IEEE International Symposium on Information
  Theory (ISIT)}, June 2015, pp. 1866--1870.

\bibitem{anyu2014combi}
A.~Wang and Z.~Zhang, ``Repair locality from a combinatorial perspective,'' in
  \emph{2014 IEEE International Symposium on Information Theory}, June 2014,
  pp. 1972--1976.

\bibitem{tamo2016bounds}
I.~Tamo, A.~Barg, and A.~Frolov, ``Bounds on the parameters of locally
  recoverable codes,'' \emph{IEEE Transactions on Information Theory}, vol.~62,
  no.~6, pp. 3070--3083, 2016.

\bibitem{huang2015linear}
P.~Huang, E.~Yaakobi, H.~Uchikawa, and P.~H. Siegel, ``Linear locally
  repairable codes with availability,'' in \emph{2015 IEEE International
  Symposium on Information Theory (ISIT)}.\hskip 1em plus 0.5em minus
  0.4em\relax IEEE, 2015, pp. 1871--1875.

\bibitem{7541336}
S.~Kadhe and A.~Sprintson, ``Codes with unequal locality,'' in \emph{2016 IEEE
  International Symposium on Information Theory (ISIT)}, July 2016, pp.
  435--439.

\bibitem{7541377}
A.~Zeh and E.~Yaakobi, ``Bounds and constructions of codes with multiple
  localities,'' in \emph{2016 IEEE International Symposium on Information
  Theory (ISIT)}, July 2016, pp. 640--644.

\end{thebibliography}

\end{document}